\documentclass[12pt]{article}
\usepackage{amsmath}
\usepackage{amssymb}
\usepackage{amsfonts}
\usepackage{latexsym}
\usepackage{color}
\usepackage{graphicx}

\catcode `\@=11 \@addtoreset{equation}{section}

\catcode `\@=12

\newcommand{\be}{\begin{equation}}
\newcommand{\en}{\end{equation}}
\newcommand{\bea}{\begin{eqnarray}}
\newcommand{\ena}{\end{eqnarray}}
\newcommand{\beano}{\begin{eqnarray*}}
\newcommand{\enano}{\end{eqnarray*}}
\newcommand{\bee}{\begin{enumerate}}
\newcommand{\ene}{\end{enumerate}}

\newcommand{\Hil}{{\cal H}}

\newcommand{\kt}{\rangle}
\newcommand{\br}{\langle}
\newcommand{\F}{{\cal F}}
\newcommand{\Lc}{{\cal L}}

\newcommand{\D}{{\cal D}}

\newcommand{\E}{{\cal E}}

\newcommand{\V}{{\cal V}}

\newcommand{\1}{1 \!\! 1}
\newtheorem{thm}{Theorem}
\newtheorem{cor}[thm]{Corollary}

\newenvironment{proof}{\noindent {\bf Proof:}}{\hfill$\Box$}

\begin{document}
\thispagestyle{empty}

\vspace*{1cm}

\begin{center}
{\Large \bf Susy for non-Hermitian Hamiltonians, with a view to coherent states}   \vspace{2cm}\\

{\large F. Bagarello}\\
\vspace*{.7cm}

\normalsize
${}^1$Dipartimento di Ingegneria, Universit\`a di Palermo, I-90128  Palermo, Italy, and \\

\vspace*{.2cm}
 INFN, Sezione di Napoli, Italy\\

\end{center}

\vspace*{0.5cm}

\begin{abstract}
\noindent We propose an extended version of supersymmetric quantum mechanics which can be useful if the Hamiltonian of the physical system under investigation is not Hermitian. The method is based on the use of two, in general different, superpotentials. Bi-coherent states of the Gazeau-Klauder type are constructed and their properties are analyzed. Some examples are also discussed, including an application to the Black-Scholes equation, one of the most important equations in Finance.
\end{abstract}

\vspace{2cm}


\vfill

\newpage

\section{Introduction}\label{sect1}

Supersymmetric quantum mechanics (Susy qm, in the following) is nowadays a well analyzed approach which has proven to be quite useful in the attempt of constructing Hamiltonians whose eigenvalues and eigenvectors can be {\em easily} deduced, out of those of a given operator. The role of factorization in this procedure is crucial, and it is widely discussed. We refer to \cite{jun,coop,suku} for many results on Susy qm and to \cite{rosas} for an interesting review on the factorization method, with a very reach list of references. The essence is the following: we consider an operator $a=\frac{d}{dx}+w(x)$, acting on $\Hil\equiv\Lc^2(\mathbb{R})$, whose adjoint is $a^\dagger=-\frac{d}{dx}+w(x)$, at least if $w(x)$ is a real function, called {\em superpotential}. Needless to say, the domains of $a$ and $a^\dagger$, $D(a)$ and $D(a^\dagger)$, cannot be all of $\Hil$, since each function in these sets must be, at least, differentiable. This suggests that, in general, they are unbounded, since all closed bounded operators can be defined everywhere in $\Hil$. For instance, if we take $w(x)$ linear in $x$ as for harmonic oscillator, it is well known that $a$ and $a^\dagger$ are unbounded. However, all throughout this paper, we will not consider in details this aspect of the operators involved in our analysis, except when it will be essential. 

Two operators can now be introduced: $h_1=a^\dagger a$ and $h_2=aa^\dagger$. In the coordinate representation, these look like:
\be
h_1=a^\dagger a=-\frac{d^2}{dx^2}+v_1(x), \qquad h_2=aa^\dagger=-\frac{d^2}{dx^2}+v_2(x),
\label{11}\en
where
\be
v_1(x)=w^2(x)-w'(x),\qquad v_2(x)=w^2(x)+w'(x).
\label{12}\en
It is easy to check that $[a,a^\dagger]=h_2-h_1=2w'(x)$, which is zero only if the superpotential is constant. Notice that $h_1$ and $h_2$ are both Hermitian and non-negative: $\langle f,h_j f\rangle\geq0$, $j=1,2$, for all $f\in D(h_j)$, the domain of $h_j$. Hence all their eigenvalues are real and non-negative. It is clear that the two vacua of $a$ and $a^\dagger$ cannot be both square-integrable. In fact, assuming that $\varphi^{(1)}(x)$ and $\varphi^{(2)}(x)$ satisfy $a\varphi^{(1)}(x)=0$ and $a^\dagger \varphi^{(2)}(x)=0$, we find that
$$
\varphi^{(1)}(x)=N_{1}\exp\left\{-\int w(x)\right\}, \qquad \varphi^{(2)}(x)=N_{2}\exp\left\{\int w(x)\right\}.
$$
We see that, if $\varphi^{(1)}(x)\in \Hil$, then $\varphi^{(2)}(x)\notin \Hil$, and vice-versa. It may also happen, however, that neither $\varphi^{(1)}(x)$ nor $\varphi^{(2)}(x)$ belong to $\Hil$. This is when SUSY is {\em broken}. In this case all the eigenvalues of $h_j$ must be strictly positive and the spectra of $h_1$ and $h_2$ coincide: $E_n^{(1)}=E_n^{(2)}=:E_n$, \cite{coop}. When SUSY is not broken ({\em unbroken} SUSY), one can always rename the operators in such a way $\varphi^{(1)}(x)\in \Hil$, while $\varphi^{(2)}(x)\notin \Hil$. This is, in fact, the standard choice adopted in the literature. In this short review, we will restrict to the broken case, since this will be the more interesting situation for us, expecially in connection with the bicoherent states considered in Section \ref{sectGKBCS}. Hence, let us assume that $e_n^{(j)}(x)$ is an eigenstate of $h_j$ with eigenvalue $E_n>0$: $h_je_n^{(j)}=E_ne_n^{(j)}$. Then
\be
e_n^{(2)}=\frac{1}{\sqrt{E_n}}\, a\,e_n^{(1)}, \qquad e_n^{(1)}=\frac{1}{\sqrt{E_n}}\, a^\dagger e_n^{(2)}.
\label{13}\en
Of course, these formulas make sense since $E_n>0$ for all $n$.
Hence $a$ and $a^\dagger$ are not, in general, ladder operators. They rather map the o.n. basis $\E_1=\{e_n^{(1)}\}$ into the second o.n. basis $\E_2=\{e_n^{(2)}\}$, and vice-versa. Notice that $\E_j$ will be assumed to be bases quite often in this paper, even if this is not always true when dealing with eigenvectors of non-Hermitian operators, see \cite{baginbagbook} for physical Hamiltonians showing this feature.

Let us introduce now the operators
$$
H_0=\left(
\begin{array}{cc}
h_1 & 0 \\
0 & h_2 \\
\end{array}
\right),\qquad
Q_0=\left(
\begin{array}{cc}
0 & 0 \\
a & 0 \\
\end{array}
\right),\qquad
Q_0^\dagger=\left(
\begin{array}{cc}
0 & a^\dagger \\
0 & 0 \\
\end{array}
\right).
$$
Then, the following formulas are satisfied:
$$
[H_0,Q_0]=[H_0,Q_0^\dagger]=0,\quad Q_0^2={Q_0^\dagger}^2=0,\quad \{Q_0,Q_0^\dagger\}=H_0.
$$
Also, if we  put 
$$
\tilde e_n^{(+)}=\left(
\begin{array}{c}
e_n^{(1)}  \\
0  \\
\end{array}
\right), \qquad \tilde e_n^{(-)}=\left(
\begin{array}{c}
0  \\
e_n^{(2)}  \\
\end{array}
\right),
$$
then
\be
H_0\tilde e_n^{(\pm)}=E_n\tilde e_n^{(\pm)}, \qquad Q_0\tilde e_n^{(+)}=\sqrt{E_n}\tilde e_n^{(-)}, \qquad Q_0^\dagger\tilde e_n^{(-)}=\sqrt{E_n}\tilde e_n^{(+)},
\label{14}\en
while $Q_0\tilde e_n^{(-)}=Q_0^\dagger\tilde e_n^{(+)}=0$. Many more details and examples of this (and similar) structure can be found in the literature on Susy qm, see \cite{jun,coop} in particular. 

In this paper we will extend this setting to the case in which an Hamiltonian $H_1$, replacing $h_1$ above, can still be factorized, but in terms of two  unrelated operators $A$ and $B$: $H_1=BA$, with $A\neq B^\dagger$. This implies, of course, that $H_1$ is not Hermitian, but opens interesting possibilities as, for instance, having zero-eigenvalue vacua for both $H_1$ and for its supersymmetric partner $H_2=AB$, as we will see. This is the content of Section \ref{sect2}. In Section \ref{sectdeformSUSY} we discuss how our framework can be deduced from ordinary SUSY using a bounded deformation operator, with bounded inverse. Section \ref{sectGKBCS} contains some preliminary results on bicoherent states of the Gazeau-Klauder type, \cite{gk}, with an application to the Swanson model, \cite{swan}. Examples are discussed in Section \ref{sectPBexample}, while our conclusions are given in Section \ref{sectcon}.

\section{The general settings}\label{sect2}

Let us consider two operators $A$ and $B$ defined as follows:
\be
A=\frac{d}{dx}+w_A(x), \qquad B=-\frac{d}{dx}+w_B(x),
\label{21}\en
with $w_A(x)$ and $w_B(x)$ in principle complex functions, and sufficiently regular\footnote{Regularity of $w_{A,B}(x)$ is required to make our computations meaningful. For instance, see (\ref{23}), they must admit at least the first derivative.}.
Of course,  if $w_A(x)=w_B(x)=w(x)$ we have $A=B^\dagger$ and these both coincide with $a$ in Section \ref{sect1} while, if $w_A(x)\neq w_B(x)$, $A$ and $B^\dagger$ are different. We still call these functions {\em superpotentials}. It should be stated clearly that, as $a$ and $a^\dagger$, also $A$ and $B$ are unbounded operators, being not everywhere defined in $\Hil$. It is now an easy computation to check that
\be
H_1=BA=-\frac{d^2}{dx^2}+q_1(x)\frac{d}{dx}+V_1(x), \qquad H_2=AB=-\frac{d^2}{dx^2}+q_1(x)\frac{d}{dx}+V_2(x),
\label{22}\en
where
\be
q_1(x)=w_B(x)-w_A(x), \quad V_1(x)=w_A(x)w_B(x)-w'_A(x), \quad V_2(x)=w_A(x)w_B(x)+w'_B(x).
\label{23}\en
It is obvious that, even if $w_{A,B}(x)$ are real functions, $H_1$ and $H_2$ are manifestly non-Hermitian, due to the presence of the term $q_1(x)\frac{d}{dx}$, which disappears only if $w_A(x)=w_B(x)$. This is exactly the situation in which $B^\dagger=A$, and we go back to ordinary Susy qm, see Section \ref{sect1}. It is also possible to {\em remove} the first derivative term by a suitable transformation of $H_1$ or $H_2$. This was discussed for the Black-Scholes equation, \cite{roy,roy2,bag2016}, where it is also shown that this transformation is implemented by an unbounded operator, with unbounded inverse. In what follows, we are more interested to considering different superpotentials since this will allow us to produce new results.

First of all, it is clear that, if we know the vacuum of $A$, i.e. the function satisfying the equation $A\varphi_0^{(1)}(x)=0$, then we can deduce the superpotential $w_A(x)$: $w_A(x)=-\frac{\frac{d}{dx}\varphi_0^{(1)}(x)}{\varphi_0^{(1)}(x)}$. Analogously, if $\varphi_0^{(2)}(x)$ is the vacuum of $B$, $B\varphi_0^{(2)}(x)=0$, then $w_B(x)=\frac{\frac{d}{dx}\varphi_0^{(2)}(x)}{\varphi_0^{(2)}(x)}$. Of course, these formulas make sense if $\varphi_0^{(1)}(x)$ and $\varphi_0^{(2)}(x)$ are never zero. Viceversa, knowing the superpotential it is possible to deduce the two vacua:
\be
\varphi_0^{(1)}(x)=N_{\varphi(1)}\,\exp\left\{-\int w_A(x)\right\}, \qquad \varphi_0^{(2)}(x)=N_{\varphi(2)}\,\exp\left\{\int w_B(x)\right\}.
\label{24}\en
Here $N_{\varphi(1)}$ and $N_{\varphi(2)}$ are two normalization constants\footnote{Calling $N_{\varphi(1)}$ and $N_{\varphi(2)}$ {\em normalization} constants could be not really appropriate, since it may happen that $\varphi_0^{(1)}(x)$ or $\varphi_0^{(2)}(x)$, or both, are not in $\Lc^2(\mathbb{R})$.}. Of course, since $w_A(x)$ and $w_B(x)$ are not necessarily connected, it may be true that both $\varphi_0^{(1)}(x)$ and $\varphi_0^{(2)}(x)$ are square integrable. For instance, if $\varphi_0^{(1)}(x)\in\Hil$,  a trivial choice which guarantees this result is $w_B(x)=-w_A(x)$. However, this is {\em too trivial}, since it implies that $B=-A$, $V_1(x)=V_2(x)$ and $H_1=H_2$. On the other hand, if we take  $w_B(x)=-\alpha w_A(x)$, for some $\alpha>0$, $\alpha\neq 1$, the situation becomes more interesting since $V_1(x)\neq V_2(x)$ and both $\varphi_0^{(1)}(x)$ and $\varphi_0^{(2)}(x)$ belong to $\Hil$. Another possible choice\footnote{Suggested by the unkonwn Referee. Thank you!} is  $w_B(x)=-\overline{w_A(x)}$, which returns a real $q_1(x)$ and $V_2(x)=\overline{V_1(x)}$.

The commutator between $A$ and $B$ is  the difference between the two Hamiltonians:
\be
[A,B]=H_2-H_1=V_2(x)-V_1(x)=w_A'(x)+w_B'(x),
\label{25}\en
which extends what deduced for $[a,a^\dagger]$. In particular, if $w_A(x)+w_B(x)$ is linear in $x$, we recover the pseudo-bosonic commutation rule, $[A,B]\propto\1$, and several interesting results can be deduced, see \cite{baginbagbook,bag,pbs} and references therein, and \cite{bgs} for a more recent results. We will consider this particular case in Section \ref{sectPBexample}. It is known that, when we deal with pseudo-bosons, $A$, $B$ and their adjoint act as ladder operators, so that the full families of eigenstates for $H_1$ and $H_2$ can be explicitly constructed in a rather automatic way, as one does for the harmonic oscillator. On the other hand, if $w_A(x)+w_B(x)$ is not linear in $x$, then this is not possible, in general, and the eigenvectors should be constructed using some alternative strategy, if any. For the moment,  we assume that, in {\em some way}, we know the eigenvectors of $H_1$ and $H_2$, and their related eigenvalues:
\be
H_1\varphi_n^{(1)}(x)=E_n^{(1)}\varphi_n^{(1)}(x), \qquad H_2\varphi_n^{(2)}(x)=E_n^{(2)}\varphi_n^{(2)}(x),
\label{26}\en
for all $n\geq0$. We are assuming also that $\varphi_0^{(1)}(x)$ and $\varphi_0^{(2)}(x)$ belong to $\Hil$, so that $E_0^{(1)}=E_0^{(2)}=E_0=0$. As for $h_1$ and $h_2$, it is possible to prove that $E_n^{(1)}=E_n^{(2)}$ for all $n\geq0$. In fact, since $\varphi_n^{(1)}(x)$ is an eigenstate of $H_1$ with eigenvalue $E_n^{(1)}$, then  $A\varphi_n^{(1)}(x)$ is  eigenstate of $H_2$ with the same eigenvalue $E_n^{(1)}$. Analogously, since  $\varphi_m^{(2)}(x)$ is an eigenstate of $H_2$ with eigenvalue $E_m^{(2)}$, then  $B\varphi_m^{(2)}(x)$ is  eigenstate of $H_1$ with the same eigenvalue $E_m^{(1)}$. Now, with a clever reordering of the eigenvalues of, say, $H_2$, we conclude that $E_n^{(1)}=E_n^{(2)}=E_n$ also for $n>0$, and that
\be
A\varphi_n^{(1)}(x)=\alpha_n\varphi_n^{(2)}(x), \qquad B\varphi_n^{(2)}(x)=\beta_n\varphi_n^{(1)}(x),
\label{27}\en
 with $\alpha_n\beta_n=E_n$, for all $n\geq1$.
 
 \vspace{2mm}
 
 {\bf Remark:--} (1) Of course, we can consider the first equation in (\ref{27}) as the {\em defining relation} for $\varphi_n^{(2)}(x)$. In other words, it is not really needed to know the eigenvectors of both $H_1$ and $H_2$. In fact, from (\ref{27}) we see that the knowledge of one family is enough to deduce also the second set.

 \vspace{1mm}

 (2) Equation (\ref{27}) also holds for $n=0$ if $E_0>0$. In this case, however, the vacua $\varphi_0^{(j)}(x)$  cannot be those in (\ref{24}), of course, since they are not compatible with the fact that we should now have, for instance, $A\varphi_0^{(1)}(x)\neq0$.
 
 \vspace{2mm}

With respect to ordinary Susy qm, we have two more Hamiltonians which are interesting for us, $H_1^\dagger$ and $H_2^\dagger$. We find
\be
H_1^\dagger=A^\dagger B^\dagger=-\frac{d^2}{dx^2}-\overline{q_1(x)}\frac{d}{dx}+\V_1(x), \qquad H_2^\dagger=B^\dagger A^\dagger=-\frac{d^2}{dx^2}-\overline{q_1(x)}\frac{d}{dx}+\V_2(x),
\label{28}\en
where
\be
\V_1(x)=\overline{w_A(x)w_B(x)-w'_B(x)}, \qquad \V_2(x)=\overline{w_A(x)w_B(x)+w'_A(x)}.
\label{29}\en
It is clear that, in general, these potentials are different from those in (\ref{23}). However, it is easy to see that each $H_j^\dagger$ has the same expression as $H_j$ with $w_{A,B}(x)$ replaced by $\overline{w_{B,A}(x)}$.
This suggests that we could repeat, for $H_1^\dagger$ and $H_2^\dagger$, what we have done for $H_1$ and $H_2$. In particular, we could look for the vacua of $A^\dagger$ and $B^\dagger$ and check under which conditions they are (both, possibly) in $\Hil$. Alternatively, we could make use of the following result, which gives us conditions for $H_j^\dagger$ to have eigenvectors, and how these functions should be related to the eigenvectors of $H_j$.

\begin{thm}\label{thm1}
	Suppose $\F_{\varphi(1)}=\{\varphi_n^{(1)}(x)\}$ is a basis of $\Hil$. Then there exist an unique set $\F_{\psi(1)}=\{\psi_n^{(1)}(x)\}$ which is also a basis of $\Hil$ and such that $\F_{\varphi(1)}$ and $\F_{\psi(1)}$ are biorthogonal. Moreover, $\psi_n^{(1)}(x)$ is eigenstate of $H_1^\dagger$ with eigenvalue $\overline{E_n}$: $H_1^\dagger \psi_n^{(1)}(x)=\overline{E_n}\,\psi_n^{(1)}(x)$. 
	A similar statement holds for $\F_{\varphi(2)}=\{\varphi_n^{(2)}(x)\}$.
\end{thm}

\begin{proof}

The existence of an  basis $\F_{\psi(1)}=\{\psi_n^{(1)}(x)\}$ which is biorthogonal to $\F_{\varphi(1)}$ is granted, see \cite{chri}.  The fact that its vectors are eigenstates of $H_1^\dagger$ is a consequence of the completeness of $\F_{\varphi(1)}$. In fact, since $\br  \psi_n^{(1)},\varphi_m^{(1)}\kt=\delta_{n,m}$,
$$
\br H_1^\dagger \psi_n^{(1)},\varphi_m^{(1)}\kt=\br \psi_n^{(1)},H_1\varphi_m^{(1)}\kt=E_m\br  \psi_n^{(1)},\varphi_m^{(1)}\kt=E_n\br  \psi_n^{(1)},\varphi_m^{(1)}\kt=\br \overline{E_n} \psi_n^{(1)},\varphi_m^{(1)}\kt.
$$
Hence $\br (H_1^\dagger-\overline{E_n}) \psi_n^{(1)},\varphi_m^{(1)}\kt=0$, for all $m$. But, since the set $\F_{\varphi(1)}$ is complete, $ (H_1^\dagger-\overline{E_n}) \psi_n^{(1)}=0$ for all $n$. Our claim follows. Of course the proof for $H_2$ is completely analogous.

\end{proof}

The counterpart of formula (\ref{27}) can be deduced also for the sets $\F_{\psi(j)}$, $j=1,2$. In particular we have 
\be
B^\dagger\psi_n^{(1)}(x)=\overline{\beta_n}\psi_n^{(2)}(x), \qquad A^\dagger\psi_n^{(2)}(x)=\overline{\alpha_n}\psi_n^{(1)}(x),
\label{210}\en
where $\alpha_n$ and $\beta_n$ are those already introduced. It is useful to draw the following picture:

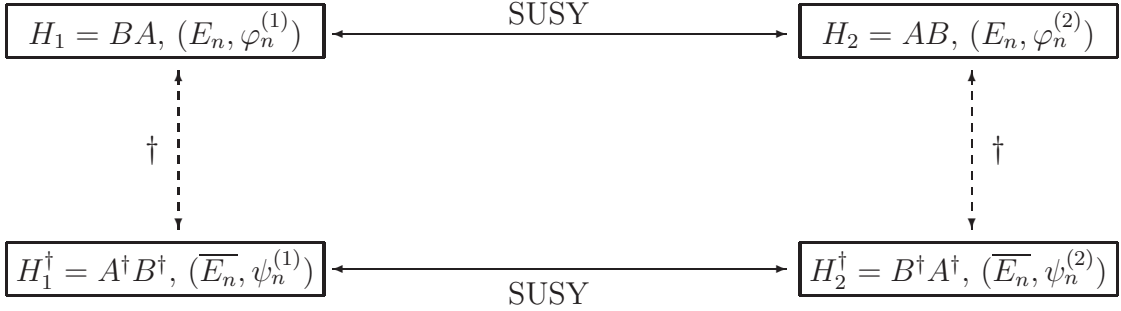
\vspace*{.1cm}


\begin{picture}(450,90)
\put(5,35){\thicklines\line(1,0){120}}
\put(5,55){\thicklines\line(1,0){120}}
\put(5,35){\thicklines\line(0,1){20}}
\put(125,35){\thicklines\line(0,1){20}}
\put(65,45){\makebox(0,0){$H_1=BA$, $(E_n,\varphi_n^{(1)})$}}

\put(305,35){\thicklines\line(1,0){120}}
\put(305,55){\thicklines\line(1,0){120}}
\put(305,35){\thicklines\line(0,1){20}}
\put(425,35){\thicklines\line(0,1){20}}
\put(365,45){\makebox(0,0){$H_2=AB$, $(E_n,\varphi_n^{(2)})$}}

\put(5,-55){\thicklines\line(1,0){120}}
\put(5,-35){\thicklines\line(1,0){120}}
\put(5,-55){\thicklines\line(0,1){20}}
\put(125,-55){\thicklines\line(0,1){20}}
\put(65,-45){\makebox(0,0){$H_1^\dagger=A^\dagger B^\dagger$, $(\overline{E_n},\psi_n^{(1)})$}}

\put(305,-55){\thicklines\line(1,0){120}}
\put(305,-35){\thicklines\line(1,0){120}}
\put(305,-55){\thicklines\line(0,1){20}}
\put(425,-55){\thicklines\line(0,1){20}}
\put(365,-45){\makebox(0,0){$H_2^\dagger=B^\dagger A^\dagger$, $(\overline{E_n},\psi_n^{(2)})$}}

\multiput(70,28)(0,-6){9}{\line(0,-1){3}}
\put(70,-27){\vector(0,-1){3}}
\put(70,27){\vector(0,1){3}}
\put(60,0){\makebox(0,0){$\dagger$}}

\multiput(370,28)(0,-6){9}{\line(0,-1){3}}
\put(370,-27){\vector(0,-1){3}}
\put(370,27){\vector(0,1){3}}
\put(380,0){\makebox(0,0){$\dagger$}}

\put(130,44){\vector(1,0){170}}
\put(130,-45){\vector(1,0){170}}
\put(130,44){\vector(-1,0){2}}
\put(130,-45){\vector(-1,0){2}}
\put(210,52){\makebox(0,0){SUSY}}
\put(210,-54){\makebox(0,0){SUSY}}

\end{picture}

\vspace*{3cm}

This diagram  shows the effects of SUSY (horizontal lines) and of the adjoint map (vertical lines): SUSY exchanges the order of the operators factorizing the various Hamiltonians mapping the "(1)" into the "(2)" sets of vectors, and vice-versa, while keeping unchanged the eigenvalues. On the other hand, the (vertical) $\dagger$ maps each Hamiltonian into its adjoint. This operation implies the replacement of the eigenvalues with their complex conjugate, and the sets $\F_{\varphi(j)}$ with $\F_{\psi(j)}$, and vice-versa.

As in ordinary SUSY we can introduce the operators
\be
H=\left(
\begin{array}{cc}
H_1 & 0 \\
0 & H_2 \\
\end{array}
\right),\qquad
Q_A=\left(
\begin{array}{cc}
0 & 0 \\
A & 0 \\
\end{array}
\right),\qquad
Q_B=\left(
\begin{array}{cc}
0 & B \\
0 & 0 \\
\end{array}
\right).
\label{211}\en
Then, 
\be
[H,Q_A]=[H,Q_B]=0,\quad Q_A^2=Q_B^2=0,\quad \{Q_A,Q_B\}=H,
\label{212}\en
with similar equalities satisfied by $H^\dagger$, $Q_A^\dagger$ and $Q_B^\dagger$. For instance, $\{Q_A^\dagger,Q_B^\dagger\}=H^\dagger$.
Also, if we further put 
$$
\tilde \varphi_n^{(+)}=\left(
\begin{array}{c}
\varphi_n^{(1)}  \\
0  \\
\end{array}
\right), \quad \tilde \varphi_n^{(-)}=\left(
\begin{array}{c}
0  \\
\varphi_n^{(2)}  \\
\end{array}
\right),\quad \tilde \psi_n^{(+)}=\left(
\begin{array}{c}
\psi_n^{(1)}  \\
0  \\
\end{array}
\right), \quad \tilde \psi_n^{(-)}=\left(
\begin{array}{c}
0  \\
\psi_n^{(2)}  \\
\end{array}
\right),
$$
we deduce that
\be
H\tilde \varphi_n^{(\pm)}=E_n\tilde \varphi_n^{(\pm)}, \qquad H^\dagger\tilde \psi_n^{(\pm)}=\overline{E_n}\,\tilde \psi_n^{(\pm)},
\label{213}\en
and
\be
Q_A\tilde \varphi_n^{(+)}=\alpha_n\tilde \varphi_n^{(-)}, \quad Q_B\tilde \varphi_n^{(-)}=\beta_n\tilde \varphi_n^{(+)},\quad Q_A^\dagger\tilde \psi_n^{(-)}=\overline{\alpha_n}\tilde \psi_n^{(+)}, \quad Q_B^\dagger\tilde \psi_n^{(+)}=\overline{\beta_n}\tilde \psi_n^{(-)},
\label{214}\en
while $Q_A\tilde \varphi_n^{(-)}=Q_B\tilde \varphi_n^{(+)}=Q_A^\dagger\tilde \psi_n^{(+)}=Q_B^\dagger\tilde \psi_n^{(-)}=0$. Formula (\ref{214}) shows how the various $Q$'s map fermionic into bosonic vectors, and vice-versa. We conclude that the essential characteristics of ordinary SUSY qm are recovered in the present setting. Still, our results look somehow richer, since the adjoint map has interesting features both for its physical consequences (the differences between, say, an Hamiltonian and its adjoint have been considered in many applications to, e.g., quantum mechanical gain and loss systems, see \cite{ben,mosta} and references therein) and from the mathematical side (many mathematical aspects of non self-adjoint operators have been considered in \cite{bagbook}). 

\section{Deformed ordinary SUSY qm}\label{sectdeformSUSY}

In this section we will show how operators like those in (\ref{21}) can be easily obtained by a suitable deformation of ordinary SUSY qm, using some kind of similarity map implemented by an invertible (but possibly non unitary) operator. This is not particularly different from what we can find in connection with some non self-adjoint Hamiltonians which are often considered in PT or pseudo-Hermitian quantum mechanics, \cite{ben,mosta}, which are deduced as (bounded or unbounded) deformations of some Hermitian operator. However, to keep the mathematical aspects of the problem under control, in what follows we will assume that the operator implementing the deformation is bounded, with bounded inverse. For that we consider a regular (at least differentiable) complex-valued function $q(x)=q_r(x)+i\,q_i(x)$, where $q_r(x)=\Re\{q(x)\}$ and $q_i(x)=\Im\{q(x)\}$, whose real part is bounded from below and from above: two strictly positive constants $m$ and $M$ exist such that
$$
0<m\leq q_r(x)\leq M<\infty.
$$
Then we define the following multiplication operator $T$, and its inverse:
\be
(Tf)(x)=e^{q(x)}f(x), \qquad (T^{-1}f)(x)=e^{-q(x)}f(x).
\label{31}\en
It is easy to check that
\be
\|T\|\leq e^M, \qquad \|T^{-1}\|\leq e^{-m}.
\label{32}\en
Therefore $f(x)$ in (\ref{31}) can be taken arbitrarily in $\Hil$. In other words, $D(T)=D(T^{-1})=\Lc^2(\mathbb{R})$. Now, if we call $A=TaT^{-1}$ and $B=Ta^\dagger T^{-1}$, where $a$ and $a^\dagger$ are those introduced in Section \ref{sect1}, we obtain exactly the operators in (\ref{21}) with
\be
w_A(x)=w(x)-q'(x), \qquad w_B(x)=w(x)+q'(x).
\label{33}\en
With this choice the Hamiltonians $H_1$ and $H_2$ in (\ref{22}) have 
\be
q_1(x)=2q'(x), \,\, V_1(x)=w^2(x)-w'(x)-(q'(x))^2+q''(x), \,\,  V_2(x)=w^2(x)+w'(x)-(q'(x))^2+q''(x),
\label{34}\en
while the potentials $\V_1(x)$ and $\V_2(x)$ in (\ref{29}) turn out to be
\be
\V_1(x)=w^2(x)-w'(x)-(\overline{q'(x)})^2-\overline{q''(x)}, \qquad  \V_2(x)=w^2(x)+w'(x)-(\overline{q'(x)})^2-\overline{q''(x)},
\label{35}\en
where we have used the fact that $w(x)$ is real. Let us now define, out of the o.n. bases $\E_1$ and $\E_2$, see Section \ref{sect1}, the following vectors and their related sets:
\be
\varphi_n^{(j)}(x)=(Te_n^{(j)})(x)=e^{q(x)}e_n^{(j)}, \qquad \psi_n^{(j)}(x)=((T^{-1})^\dagger e_n^{(j)})(x)=e^{-\overline{q(x)}}e_n^{(j)},
\label{36}\en
and $\F_{\varphi(j)}=\{\varphi_n^{(j)}(x)\}$, $\F_{\psi(j)}=\{\psi_n^{(j)}(x)\}$, $j=1,2$. Due to the boundedness of $T$ and $T^{-1}$, each pair $(\F_{\varphi(j)}, \F_{\psi(j)})$ is a biortogonal Riesz basis, the best we can have after o.n. bases, \cite{chri}. Of course, this would not be true if $q_r(x)$ does not satisfy the inequalities given at the beginning of this section. In particular, as a consequence of (\ref{32}), we deduce that
\be
\|\varphi_n^{(j)}\|\leq e^M, \qquad \|\psi_n^{(j)}\|\leq e^{-m}, 
\label{37}\en
$j=1,2$.
Moreover, they are eigenstates of the various Hamiltonians we have introduced so far. More in detail,
\be
H_j\varphi_n^{(j)}(x)=E_n\varphi_n^{(j)}(x), \qquad H_j^\dagger\psi_n^{(j)}(x)=E_n\psi_n^{(j)}(x),
\label{38}\en
$j=1,2$. Notice that the eigenvalues are all real and, therefore, coincident, even if the Hamiltonians are manifestly non Hermitian. This is because each $H_j$ is defined as a deformation of an Hermitian operator, $h_j$, whose eigenvalues are necessarily real and non negative. This aspect was already commented in Section \ref{sect1}.

Straightforward computations show that, for instance, equations (\ref{27}) and (\ref{210}) are satisfied with the choice $\alpha_n=\beta_n=\sqrt{E_n}$. Also, if we introduce the matrix
$$
T_2=\left(
\begin{array}{cc}
T & 0 \\
0 & T \\
\end{array}
\right),
$$
then $T_2$ is invertible and the following equalities between the quantities introduced in Sections \ref{sect1} and \ref{sect2} hold:
$$
H=T_2H_0T_2^{-1}, \quad Q_A=T_2Q_0T_2^{-1}, \quad Q_B=T_2Q_0^\dagger T_2^{-1}, \quad \tilde \varphi_n^{(\pm)}=T_2 \tilde e_n^{(\pm)}\quad \tilde \psi_n^{(\pm)}=(T_2^{-1})^\dagger \tilde e_n^{(\pm)}.
$$
Hence there exists a sort of equivalence between standard SUSY qm and what has been deduced in Section \ref{sect2}, at least under our assumptions on $q_r(x)$. It is well known, however, that this is not really so evident if $T$,  $T^{-1}$, or both, are unbounded. This is the case if $q_r(x)$ is not bounded but is only, for instance, semi-bounded. In this case $\F_{\varphi(j)}$, $\F_{\psi(j)}$ are not Riesz bases, and may also not be even bases for $\Hil$. We refer to \cite{baginbagbook,bagbell} for some results on this aspect in the context of non-Hermitian Hamiltonians. In the next Section we will show how the vectors in (\ref{36}) can be used to introduce a certain class of bi-coherent states.

\section{Gazeau-Klauder-like bicoherent states}\label{sectGKBCS}

In this section we will show how some sort of coherent states can be naturally {\em attached} to the general framework discussed so far. In particular, we will significantly extend what we have done in \cite{bag2008}, where the idea of using a particular definition of coherent states, due to Gazeau and Klauder \cite{gk}, was already considered. It is maybe worth stressing that  many kind of coherent states have been introduced and studied during the years, with different features and related to different physical systems. We refer to the following monographs, \cite{gazeaubook}-\cite{aag}, and to the  recent volume \cite{cpcs}. In our knowledge, all the coherent states proposed so far are  eigenvectors of some lowering operator and satisfy some sort of resolution of the identity. In recent years, with the growing interest for non-Hermitian Hamiltonians, many attempts to define coherent states also in this case have been carried out, see the recent paper \cite{deyetc} for instance. In particular, we have proposed one of these extentions, the so-called bi-coherent states, see \cite{bgs} and references therein, which have a lot of nice properties. Here, extending what we did in \cite{bag2008}, we propose a different kind of bi-coherent states which we call {\em of the Gazeau-Klauder type}, \cite{gk}, which are rather different from those proposed in \cite{fhro}-\cite{fring}.

Let $\F_{\varphi(j)}=\{\varphi_n^{(j)}(x)\}$ be a basis\footnote{The reader could have in mind the sets in Theorem \ref{thm1}. However, most of what we will discuss in this section holds true independently of the origin of these vectors.} for $\Hil$ such that $H_j\varphi_n^{(j)}(x)=E_n\varphi_n^{(j)}(x)$. Let then $\F_{\psi(j)}=\{\psi_n^{(j)}(x)\}$ be the unique biorthogonal basis associated to $\F_{\varphi(j)}$, $j=1,2$, \cite{chri}. We have already shown in Theorem \ref{thm1} that $H_j^\dagger\psi_n^{(j)}(x)=\overline{E_n}\psi_n^{(j)}(x)$. Let us now define the following vectors
\be
\varphi^{(j)}(J,\gamma;x)=K(J)\,\sum_{n=0}^\infty \frac{J^{n/2}e^{-iE_n\gamma}}{\sqrt{\rho_n}}\varphi_n^{(j)}(x), \quad 
\psi^{(j)}(J,\gamma;x)=K(J)\,\sum_{n=0}^\infty \frac{J^{n/2}e^{-i\overline{E_n}\gamma}}{\sqrt{\rho_n}}\psi_n^{(j)}(x),
\label{41}\en
where $J\geq 0$, $\gamma\in\mathbb{R}$, and $\rho_n$ is defined as follows: $\rho_0=1$ and $\rho_n=E_1 E_2\cdots E_n$. Since these quantities are, in general, complex, we need to clarify what we mean for $\sqrt{\rho_n}$. We make the easiest choice: if $\rho_n=|\rho_n|e^{i\theta_n}$, then $\sqrt{\rho_n}=|\rho_n|^{1/2}e^{i\theta_n/2}$. Of course, there is no problem at all if the sets  $\F_{\varphi(j)}$ and  $\F_{\psi(j)}$ are those introduced in Section \ref{sectdeformSUSY}, since all the $E_n$ are non negative. In (\ref{41}) we have also introduced the following normalization function:
\be
K(J)=\left(\sum_{n=0}^{\infty}\frac{J^n}{|\rho_n|}\right)^{-1/2},
\label{42}\en
which we take coincident for $\varphi^{(j)}(J,\gamma;x)$ and for $\psi^{(j)}(J,\gamma;x)$. The series converge if $J<R$, where $R=\lim_{n,\infty}|E_{n}|$. Incidentally we observe that $K$ does not depend on $\gamma$ or $x$. It only depends on $J$. Of course, taking $J<R$, does not ensure us that also the two series in (\ref{41}) converge. In fact, something else should be assumed on the norms of $\varphi_n^{(j)}(x)$ and $\psi_n^{(j)}(x)$. We adopt here the same very mild assumptions considered in \cite{bgs2}: let us suppose that four strictly positive constants exist, $A_\varphi$,  $A_\psi$,  $r_\varphi$,  $r_\psi$, together with two sequences $\{M_n(\varphi)\}$ and $\{M_n(\psi)\}$, such that
$$
\lim_{n,\infty}\frac{M_n(\varphi)}{M_{n+1}(\varphi)}=M(\varphi)\in]0,\infty], \qquad \lim_{n,\infty}\frac{M_n(\psi)}{M_{n+1}(\psi)}=M(\psi)\in]0,\infty],
$$
and
\be
\|\varphi_n^{(j)}\|\leq A_\varphi\,r_\varphi^n M_n(\varphi), \qquad\qquad \|\psi_n^{(j)}\|\leq A_\psi\,r_\psi^n M_n(\psi),
\label{43}\en
for all $n\geq 0$. Suppose further that, calling $E_n=E_n^{(r)}+i\,E_n^{(i)}$, the following holds:
\be
\delta E=\lim_{n,\infty}(E_n^{(i)}-E_{n+1}^{(i)})=0.
\label{44}\en
Of course, this is true if $E_n$ is real (at least for $n$ large enough), or if the imaginary part of $E_n$ is constant (up to at most a finite number of $n$), or yet if the sequence $\{E_n^{(i)}\}$ decays to zero. Then we have, for instance,
$$
\|\varphi^{(j)}\|\leq K(J) A_\varphi \,\sum_{n=0}^\infty \frac{M_n(\varphi)e^{E_n^{(i)}\gamma}}{\sqrt{|\rho_n|}}(\sqrt{J}\,r_\varphi)^n,
$$
where we have used the fact that $K(J)$ is positive. The series on the right-hand side converges if $J<J_\varphi:=\frac{M^2(\varphi)R}{r_\varphi}$, independently of $\gamma$. Analogously, the series for $\|\psi^{(j)}\|$ converges for all $\gamma$ if $J<J_\psi:=\frac{M^2(\psi)R}{r_\psi}$. Hence we can conclude that the vectors in (\ref{41}) are well defined for all $\gamma$, if $J<J_{min}=\min(R,J_\varphi,J_\psi)$. 

Now that we know the domain in which these states are defined, we are interested in deducing their properties. In the following we will call $C$ the following subset of $\mathbb{R}^2$: $C=\{(J,\gamma): J\in[0,J_{min}[, \,\gamma \in \mathbb{R}\}$.

First of all, a direct computation shows that, thanks to our choice of $K(J)$,
\be
\br \varphi^{(j)}(J,\gamma;x),\psi^{(j)}(J,\gamma;x)\kt=1,
\label{45}\en
for all $(j,\gamma)\in C$. This is a direct consequence of the biorthogonality of the families $\F_{\varphi(j)}$ and $\F_{\psi(j)}$. This kind of {\em normalization in pairs} is typical of biorthogonal sets, \cite{baginbagbook,chri}. Now, following \cite{gk}, we introduce the following measure on $C$: $d\nu(J,\gamma)=K^{-2}(J)\rho(J)dJ\,d\nu(\gamma)$, where $\rho(J)$ is a solution of the moment problem
\be
\int_{0}^{J_{min}}J^n\,\rho(J)\,dJ=|\rho_n|,
\label{46}\en
while $d\nu(\gamma)$ is defined as follows, \cite{gk}: $$\int_{\mathbb{R}} \cdots d\nu(\gamma)=\lim_{\Gamma,\infty}\frac{1}{2\Gamma}\int_{-\Gamma}^{\Gamma}\cdots d\gamma.$$
Then, if each $E_n$ has multiplicity one, it is possible to check that
$$
\int_{C}d\nu(J,\gamma)\br f, \varphi^{(j)}(J,\gamma;x) \kt\br  \psi^{(j)}(J,\gamma;x),g\kt=$$
\be =\int_{C}d\nu(J,\gamma)\br f, \psi^{(j)}(J,\gamma;x) \kt\br  \varphi^{(j)}(J,\gamma;x),g\kt=\br f,g\kt,
\label{47}\en
for all $f,g\in\Hil$. Hence the two families in (\ref{41}) resolve the identity. 

Another useful property of these states can be deduced if $e^{-iH_jt}$ and $e^{-iH_j^\dagger t}$ commute with the series in (\ref{41}). In this case, in fact, we deduce that 
\be
e^{-iH_jt}\varphi^{(j)}(J,\gamma;x)=\varphi^{(j)}(J,\gamma+t;x), \qquad e^{-iH_j^\dagger t}\psi^{(j)}(J,\gamma;x)=\psi^{(j)}(J,\gamma+t;x).
\label{48}
\en
This means that our Gazeau-Klauder-like bicoherent states are stable under time evolution. Last but not least, they also satisfy the following generalized version of the action identity, \cite{gk}, at least if $E_0=0$ and $E_n>0$ for $n>0$:
\be
\br \psi^{(j)}(J,\gamma;x),H_j\varphi^{(j)}(J,\gamma;x)\kt=J.
\label{49}\en
On the negative side, it is not a big surprise the fact that these states are not eigenstates of any of  the operators $A$, $B$, $A^\dagger$ or $B^\dagger$. In fact, this does not even hold for the standard Gazeau-Klauder coherent states. The reason is simple: except that for some particular situation, these operators are not at all lowering operators. This is true for pseudo-bosons, \cite{baginbagbook}, but not in general. However, as in \cite{gk}, some $\gamma$-depending lowering operators can be defined, via their action on the bases $\F_{\varphi(j)}$ and $\F_{\psi(j)}$. For instance, if we put
\be
a_j(\gamma)\varphi_n^{(j)}=\left\{
\begin{array}{ll}
	0 \hspace{45mm}\mbox{ if }  n=0\\
	\sqrt{E_n}\,e^{i(E_n-E_{n-1})\gamma}\varphi_{n-1}^{(j)}\hspace{8mm}\mbox{ if } n\geq 1,\\
\end{array}
\right.\label{410}\en
 we find that
\be
a_j(\gamma)\varphi^{(j)}(J,\gamma;x)=\sqrt{J}\,\varphi^{(j)}(J,\gamma;x),
\label{411}\en
$j=1,2$. Similarly, if we define 
\be
b_j^\dagger(\gamma)\psi_n^{(j)}=\left\{
\begin{array}{ll}
	0 \hspace{45mm}\mbox{ if }  n=0\\
	\sqrt{E_n}\,e^{i(\overline{E_n}-\overline{E_{n-1}})\gamma}\psi_{n-1}^{(j)}\hspace{8mm}\mbox{ if } n\geq 1,\\
\end{array}
\right.\label{412}\en
then
\be
b_j^\dagger(\gamma)\psi^{(j)}(J,\gamma;x)=\sqrt{J}\,\psi^{(j)}(J,\gamma;x),
\label{413}\en
$j=1,2$. 

It is worth noticing that in (\ref{412}) only the $E_n$'s in  the exponent  are replaced by their complex conjugates, while $\sqrt{E_n}$ appears both in (\ref{410}) and in (\ref{412}). This is related to the definitions in (\ref{41}) where we have used $\sqrt{\rho_n}$ in the denominator for both vectors. This was meant to simplify the computations, by getting $|\rho_n|$ rather that $\sqrt{\rho_n^2}=\rho_n$, for instance, when solving the moment problem (\ref{46}). In fact, in this case, we should have a complex $\rho(J)$, which we prefer to avoid.  

\vspace{2mm}

{\bf Remarks:--} (1) Even if $A$ and $B$, in general, act on our  bicoherent states in a rather complicated way, there are few situations in which interesting formulas can be deduced. This is, for instance,
 when $\alpha_n=E_n$ and $\beta_n=1$, for all $n$ in (\ref{27}). In this case we find that
 \be
 A\varphi^{(1)}(J,\gamma;x)= i\,\frac{d}{d\gamma}\varphi^{(2)}(J,\gamma;x), \qquad  B\varphi^{(2)}(J,\gamma;x)= \varphi^{(1)}(J,\gamma;x).
 \label{414}\en
 If we rather have $\alpha_n=1$ and $\beta_n=E_n$, we get 
\be
A\varphi^{(1)}(J,\gamma;x)= \varphi^{(2)}(J,\gamma;x), \qquad  B\varphi^{(2)}(J,\gamma;x)= i\,\frac{d}{d\gamma} \varphi^{(1)}(J,\gamma;x).
\label{415}\en

\vspace{2mm}

(2) In \cite{bgs,bgs2} we have considered the (1-d and 2-d) Swanson model, whose Hamiltonian is, in its 1-d version,
$$
H_{\theta}=\frac{1}{2\cos(2\theta)}\left(\hat p^2e^{-2i\theta}+\hat q^2e^{2i\theta}\right).$$ %
Here  $\hat q$ and $\hat p$ are the position and momentum operators, and $\theta$ is a real parameter taking values in $(-\frac{\pi}{4},\frac{\pi}{4})\backslash\{0\}$. It is clear that $H_\theta\neq H_\theta^\dagger$, and it is known that it can be written, except that for a  constant, in a factorized form. In fact, introducing the pair of pseudo bosonic operators
\be
\hat a=\frac{1}{\sqrt{2}}\left(\hat q_0 e^{i\theta}+i\hat p_0e^{-i\theta}\right),\quad \hat b=\frac{1}{\sqrt{2}}\left(\hat q_0 e^{i\theta}-i\hat p_0e^{-i\theta}\right),
\label{sw1},\en
see \cite{baginbagbook}, they satisfy
\be
[\hat a,\hat b]=\1,\quad \hat a^\dag\neq \hat b,
\label{sw2}\en
and moreover
$$H_\theta=\frac{1}{\cos(2\theta)}\left(\hat b \hat a+\frac{1}{2}\1\right).$$
As shown in \cite{baginbagbook} the  eigenstates of $H_\theta$ and $H_\theta^\dagger$ are respectively
\be
\varphi_n(x)=\frac{N_1}{\sqrt{2^nn!}}H_n(e^{i\theta}x)\textrm{exp}\left\{-\frac{1}{2}e^{2i\theta}x^2\right\},\quad\Psi_n(x)=\frac{N_2}{\sqrt{2^nn!}}H_n(e^{-i\theta}x)\textrm{exp}\left\{-\frac{1}{2}e^{-2i\theta}x^2\right\}
\en
for all $n\geq0$, with $N_1\bar{N}_2=\frac{e^{-i\theta}}{\sqrt{\pi} }$ in order to have $\left<\varphi_{ 0},\Psi_0\right>=1$. In \cite{bgs} we have shown that the bounds in (\ref{43}) (for $j=1$) are satisfied and that $J_{min}=\infty$. Therefore, $\varphi^{(1)}(J,\gamma;x)$ and $\psi^{(1)}(J,\gamma;x)$ are well defined for all $(J,\gamma)\in\mathbb{R}^2$. As for the other pair, $\varphi^{(2)}(J,\gamma;x)$ and $\psi^{(2)}(J,\gamma;x)$, we cannot apply our previous results as they are, since the Susy partner of $H_\theta$, which is essentially $\hat a\hat b$, has eigenvalues which are shifted with respect to those of $H_\theta$. 

\vspace{2mm}

(3) Examples of Gazeau-Klauder-like bicoherent states can be easily constructed by any broken Susy, adopting the  approach 
considered in Section \ref{sectdeformSUSY}. The inequalities in (\ref{37}) ensure the validity of those in (\ref{43}) with $A_\varphi=e^M$, $A_\psi=e^{-m}$, and $r_\varphi=r_\psi=M_n(\varphi)=M_n(\psi)=1$, $\forall n$. Hence $J_{min}=R=\lim_{n,\infty}E_n$, and for those problems for which the moment problem can be solved, the resolution of the identity in (\ref{47}) follows. Examples of broken Susy can be found, for instance, in \cite{jun,jun2,jun3}

\section{Examples}\label{sectPBexample}

As we have anticipated, if the operators $A$ and $B$ in (\ref{21}) are pseudo-bosonic, \cite{baginbagbook}, they can also be used to construct explicitly the different families of eigenvectors considered all along this paper. The price to pay, however, is that the eigenvalues of $H_j$ and $H_j^\dagger$ are essentially linear in the quantum number, since these Hamiltonians are, in general, directly proportional to certain pseudo-bosonic number operators. In this section we will consider an examples of this kind, and an interesting generalization of it. However, before going to the {\em simple} pseudo-bosonic settings, we will discuss some results related to the Black-Scholes equation, for which pseudo-bosonic ladder operators cannot be introduced.

\subsection{The Black-Scholes Hamiltonian}

The starting point is  the Black-Scholes equation for option pricing with constant volatility $\sigma$,
\be
\frac{\partial C}{\partial t}=-\frac{1}{2}\sigma^2 S^2 \frac{\partial^2 C}{\partial S^2}-rS\frac{\partial C}{\partial S}-rC.
\label{51}
\en
Here $C(S,t)$ is the price of the option, $S$ is the stock price and $r$ is the risk-free spot interest rate. Equation (\ref{51}) describes the price of an option in absence of uncertainty, when there is no random fluctuations ( {\em perfectly hedged portfolio}). A full understanding of the meaning of this equation is outside the scopes of this paper.  We refer to \cite{bs,baaquie} for an overview on (\ref{51}), its derivation and its role in Finance, and to \cite{roy2,baaquie,bag2016} for its connections with Quantum Mechanics.

 Introducing a new variable $x$ via $S=e^x$, and the related unknown function $\Psi(x)$ as $C(S(x),t)=e^{\epsilon t}\Psi(x)$, equation (\ref{51}) can be rewritten in the following form:
\be
H_{BS}\Psi(x)=\epsilon\Psi(x), \qquad \mbox{where}\qquad H_{BS}=-\frac{1}{2}\sigma^2  \frac{d^2 }{d x^2}+\left(\frac{\sigma^2}{2}-r\right)\frac{d}{dx}+r\,\1.
\label{52}\en
Notice that $H_{BS}\neq H_{BS}^\dagger$, because of the presence of the term $\left(\frac{\sigma^2}{2}-r\right)\frac{d}{dx}$. In \cite{roy,roy2,baaquie,bag2016} this Hamiltonian has been mapped into a formally Hermitian operator, using a suitable similarity multiplication operator which, contrarily to $T$ and $T^{-1}$ in (\ref{31}), is unbounded with unbounded inverse. This kind of transformation is not what we are interested in, here. In fact, $H_{BS}$ as it is, looks exactly as the Hamiltonians in (\ref{22}). For concreteness, we fix $\sigma^2=2$ to simplify the notation, and we call $H_1$ the Hamiltonian we get in this way:
\be
H_1=-\frac{d^2}{dx^2}+(1-r)\frac{d}{dx}+r\1.
\label{53}\en
Our aim is to use our results to analyse $H_1$, and its related Hamiltonians. Comparing (\ref{53}) with (\ref{22}) we see that $q_1(x)=1-r=w_B(x)-w_A(x)$, and $V_1(x)=r=w_A(x)w_B(x)-w'_A(x)$. These equations return $w_A(x)$ and $w_B(x)$ and the solution depends on whether $r=-1$ or not. In particular, if $r\neq-1$, calling
$$
v(x)=v_0e^{-(r+1)x}
-\frac{1}{r+1},$$
we find
\be
w_A(x)=r+\frac{1}{v(x)}, \qquad w_B(x)=1+\frac{1}{v(x)}.
\label{54}\en
Here $v_0$ is an integration constant, which we always take strictly positive, to fix the ideas. As for $V_2(x)$ we find that
\be
V_2(x)=\frac{1}{v^2(x)}\left(rv^2(x)+2v(x)(r+1)+2\right),
\label{55}\en
which diverges when $x$ approaches $x_0:=\frac{1}{r+1}\log((r+1)v_0)$. We see that $V_1(x)$ and $V_2(x)$ look rather different. However, as $|x|$ diverges, $V_2(x)$ converges to $r$. Hence, asymptotically, $V_1(x)$ and $V_2(x)$ are, in fact, not so different.

Let us now see what happens if $r=-1$. In this case we still have (\ref{54}), but $v(x)=v_0-x$. As before, $\lim_{|x|,\infty}V_2(x)=-1=r=V_1(x)$, and the two potentials are asymptotically quite close. This is because we find that $V_2(x)=\frac{2}{(x-v_0)^2}-1$. What we want to do now is to look for the vacua of the various operators introduced in Section \ref{sect2}, and to check whether is possible that they are all in $\Lc^2(\mathbb{R})$. We first recall that
$$
\varphi_0^{(1)}(x)=N_{\varphi(1)}\exp\left\{-\int w_A(x)dx \right\}, \qquad \varphi_0^{(2)}(x)=N_{\varphi(2)}\exp\left\{\int w_B(x)dx \right\}.
$$
Now, if $r\neq-1$,
$$
-\int w_A(x)dx=-rx+\log\left|e^{(r+1)x}-v_0(r+1)\right|,
$$
while
$$
\int w_B(x)dx=x-\log\left|e^{(r+1)x}-v_0(r+1)\right|.
$$
On the other hand, if $r=-1$ we get
$$
-\int w_A(x)dx=x+\log\left|x-v_0\right|,
\qquad 
\int w_B(x)dx=x-\log\left|x-v_0\right|.
$$
If $r>-1$, it is possible to check that $\varphi_0^{(1)}(x)$ behaves as $e^x$ for $x\rightarrow\infty$ and as $e^{-rx}$ for $x\rightarrow-\infty$. Hence $\varphi_0^{(1)}(x)\notin\Lc^2(\mathbb{R})$. On the other hand, $\varphi_0^{(2)}(x)$ behaves as $e^{-rx}$ for $x\rightarrow\infty$ and as $e^{x}$ for $x\rightarrow-\infty$. This implies that, if $r>0$ (then $r>-1$ as well, of course), $\varphi_0^{(2)}(x)\in\Lc^2(\mathbb{R})$. Of course, since $B\varphi_0^{(2)}(x)$ is proportional to $\varphi_0^{(1)}(x)$, see (\ref{27}), and since $\varphi_0^{(1)}(x)\notin\Lc^2(\mathbb{R})$, our result implies that $\varphi_0^{(2)}(x)\notin D(B)$, the domain of $B$. 

If $r=-1$ it is easy to check that both  $\varphi_0^{(1)}(x)$ and  $\varphi_0^{(2)}(x)$ behave as $e^x$, for large $|x|$. This implies that none of these vacua belong to $\Lc^2(\mathbb{R})$.

Finally, if $r<-1$, $\varphi_0^{(1)}(x)$ goes as $e^x$ for $x\rightarrow-\infty$ and as $e^{-rx}$ for $x\rightarrow+\infty$, while $\varphi_0^{(2)}(x)$ behaves as $e^{-rx}$ for $x\rightarrow-\infty$ and as $e^{x}$ for $x\rightarrow+\infty$. Then, since $r<-1$, neither  $\varphi_0^{(1)}(x)$ nor  $\varphi_0^{(2)}(x)$ are square integrable.

Let us now consider the adjoint Hamiltonians $H_1^\dagger$ and $H_2^\dagger$. First of all it is clear that, taking $v_0\in\mathbb{R}$, $w_A(x)$ and $w_B(x)$ are real. Therefore $q_1(x)=\overline{q_1(x)}$, $V_1(x)=\V_1(x)$ and $V_2(x)=\V_2(x)$. However, this does not mean that $H_j=H_j^\dagger$, of course, due to the different sign in the term linear in the $x$-derivative, see (\ref{22}) and (\ref{28}). The ground states  of $H_1^\dagger=A^\dagger B^\dagger$ and $H_2^\dagger=B^\dagger A^\dagger$ can be deduced looking for the vacua of $B^\dagger$ and $A^\dagger$ respectively: $B^\dagger\psi_0^{(1)}(x)=0$ and $A^\dagger\psi_0^{(2)}(x)=0$ , and we find
$$
\psi_0^{(1)}(x)=N_{\psi(1)}\exp\left\{-\int w_B(x)dx \right\}, \qquad \psi_0^{(2)}(x)=N_{\psi(2)}\exp\left\{\int w_A(x)dx \right\},
$$
which can be easily related to $\varphi_0^{(1)}(x)$ and $\varphi_0^{(2)}(x)$. Indeed we find
$$
\psi_0^{(1)}(x)\varphi_0^{(2)}(x)=N_{\psi(1)}N_{\varphi(2)}, \qquad \psi_0^{(2)}(x)\varphi_0^{(1)}(x)=N_{\psi(2)}N_{\varphi(1)}.
$$
Now, to analyze their asymptotic behaviour, we first consider the case $r>-1$. Of course, since $\varphi_0^{(2)}(x)\rightarrow 0$ for $x\rightarrow-\infty$,  $\psi_0^{(1)}(x)\rightarrow \infty$ in the same limit. Hence, independently of the choice of $r$, $\psi_0^{(1)}(x)\notin\Lc^2(\mathbb{R})$. The situation changes for $\psi_0^{(2)}(x)$. In fact, as we have seen,  $\varphi_0^{(1)}(x)\rightarrow \infty$ for $|x|\rightarrow\infty$, if $r>0$. Then $\psi_0^{(2)}(x)\rightarrow 0$ in the same limit. The conclusion is the following:

if $r>0$ the ground states of $H_2$ and $H_2^\dagger$, $\varphi_0^{(2)}(x)$ and $\psi_0^{(2)}(x)$, are square integrable. On the other hand, for any $r>-1$, neither $\varphi_0^{(1)}(x)$ nor $\psi_0^{(1)}(x)$ belong to $\Lc^2(\mathbb{R})$. 

A similar analysis can be repeated for the other values of $r$: if $r\leq-1$, none of the $\psi_0^{(j)}(x)$'s is square-integrable.

Summarizing, if we want our system to live in $\Lc^2(\mathbb{R})$ (assumption which, however, could be relaxed, see \cite{bgr,bagweak}), we cannot work directly with the Black-Scholes Hamiltonian but, rather, with its Susy partner $H_2$, together with the adjoint of $H_2$, at least  if $r>0$. Of course, our analysis should be refined if we also want to analyze  eigenvectors 
different from the lowest ones. However, this analysis is much harder, since there is no general way to construct them. In particular, the commutator between $A$ and $B$ is not so easy and does not allow to use them as ladder operators.

\subsection{Pseudo-bosonic superpotentials}

In this example we will discuss a non-trivial choice of $w_A(x)$ and $w_B(x)$ such that $[A,B]=\1$, see (\ref{25}). In particular, we want the two functions to be different and not just linear in $x$ (this choice has been considered several times in recent years, see \cite{baginbagbook} and references therein). To be concrete, we take here
\be
w_A(x)=k+e^{x}, \qquad w_B(x)=x-e^x,
\label{56}\en
where we take $k\in\mathbb{R}$ for concreteness. Hence
$$
A=\frac{d}{dx}+k+e^x, \qquad B=-\frac{d}{dx}+x-e^x.
$$
 With this choice we have $w_A'(x)+w_B'(x)=1$. Therefore $[A,B]=\1$, and
\be
q_1(x)=x-k-2e^x, \quad V_1(x)=kx+(x-k-1)e^x-e^{2x}, \quad V_2(x)=V_1(x)+1,
\label{57}\en
in agreement with the fact that $H_2=AB=[A,B]+BA=H_1+\1$. Needless to say, a similar relation is also recovered for $H_1^\dagger$ and $H_2^\dagger$: $H_2^\dagger=H_1^\dagger+\1$. It is clear that $\varphi_n^{(2)}(x)$ coincide with  $\varphi_n^{(1)}(x)$, but the eigenvalues are shifted. Similarly, $\psi_n^{(2)}(x)$ coincide with  $\psi_n^{(1)}(x)$, and again the eigenvalues are shifted. This reminds very much what happens in unbroken SUSY. 

\vspace{2mm}

{\bf Remark:--} This result does not really depend on the particular choice of the superpotential in (\ref{56}). Indeed, it is  a consequence of the fact that $[A,B]=\1$, which implies that $(H_2,H_2^\dagger)=(H_1+\1,H_1^\dagger+\1)$. For this reason in the rest of this example we will concentrate on what happens for $(H_1,H_1^\dagger)$, since this completely determines what happens also for the other pair, $(H_2,H_2^\dagger)$.

\vspace{2mm}

The ground state of $H_1$ can be deduced by solving the equation $A\varphi_0^{(1)}(x)=0$: $\varphi_0^{(1)}(x)=N_{\varphi(1)}e^{-kx-e^x}$. It is clear that, if $k<0$, $\varphi_0^{(1)}(x)\in\Lc^2(\mathbb{R})$. Since $B^\dagger=\frac{d}{dx}+x-e^x$, its vacuum turns out to be  $\psi_0^{(1)}(x)=N_{\psi(1)}e^{e^x-x^2/2}$, which is not in $\Lc^2(\mathbb{R})$. However, it is clear that $\psi_0^{(1)}(x)\in\Lc^1_{loc}(\mathbb{R})$, the set of all locally integrable functions, and that it is also a $C^\infty$ function: it admits derivatives of all order. It is also clear that the product $\varphi_0^{(1)}(x)\psi_0^{(1)}(x)$ belongs to $\Lc^1(\mathbb{R})$, so that the ordinary scalar product between two $\Lc^2$ functions can also be defined for this product. In fact, we will see later that this is true for all products of the kind $\varphi_k^{(1)}(x)\psi_l^{(1)}(x)$, where $\varphi_k^{(1)}(x)$ are the eigenstates of $H_1$ (all in $\Lc^2(\mathbb{R})$) while, $\psi_l^{(1)}(x)$ are the (generalized) eigenstates of $H_1^\dagger$ (none in $\Lc^2(\mathbb{R})$). This is what happens, in general, in the so-called PIP-spaces\footnote{Here PIP stands for partial inner product.}, \cite{anttra}. This reflects the general structure of {\em weak pseudo-bosons}, recently introduced in \cite{bagweak}.

Motivated by our results in \cite{baginbagbook,bagweak}, we can prove the following:

\begin{thm}
If we put
\be
\varphi_n^{(1)}(x)=\frac{1}{\sqrt{n!}}B^n\varphi_0^{(1)}(x), \qquad \psi_n^{(1)}(x)=\frac{1}{\sqrt{n!}}{A^\dagger}^n\psi_0^{(1)}(x),
\label{58}\en
$n=1,2,3,\ldots$, we have 
\be
\frac{\varphi_n^{(1)}(x)}{\varphi_0^{(1)}(x)}=\frac{\psi_n^{(1)}(x)}{\psi_0^{(1)}(x)}=p_n(x),
\label{59}\en
$\forall n\geq0$, where $p_n(x)$ is defined recursively as follows:
\be
p_0(x)=1, \qquad p_{n}(x)=\frac{1}{\sqrt{n}}\left(p_{n-1}(x)(x+k)-p_{n-1}'(x)\right).
\label{510}\en
Moreover,  $\varphi_n^{(1)}(x)\in\Lc^2(\mathbb{R})$, while  $\psi_n^{(1)}(x)\in\Lc^1_{loc}(\mathbb{R})$, $\forall n\geq0$.

\end{thm}

\begin{proof}
To prove (\ref{59}) we use induction on $n$. Of course, the statement is true for $n=0$. Let us assume that it is also true for a given $n$, with $p_n(x)$ as in (\ref{510}), and let us tehn prove that the analogous result holds for $n+1$. 

Due to (\ref{58}) and to the definition of $B$ we have
$$
\sqrt{n+1}\,\varphi_{n+1}^{(1)}(x)=B\varphi_n^{(1)}(x)=\left(-\frac{d}{dx}+x-e^x\right)\varphi_n^{(1)}(x)=\left(-\frac{d}{dx}+x-e^x\right)\left(p_n(x)\varphi_0^{(1)}(x)\right),
$$
where we have also used the induction hypothesis. Now, since $A\varphi_0^{(1)}(x)=0$, we have $\frac{d}{dx}\varphi_0^{(1)}(x)=-(k+e^x)\varphi_0^{(1)}(x)$. Therefore, after some computations, 
$$
\sqrt{n+1}\,\varphi_{n+1}^{(1)}(x)=\left(p_n(x)(x+k)-p_n'(x)\right)\varphi_0^{(1)}(x)=\sqrt{n+1}\,p_{n+1}(x)\varphi_0^{(1)}(x),
$$
which is what we had to prove. A similar proof can be repeated for $\psi_n^{(1)}(x)$. As for the nature of the functions, this is an obvious consequence of (\ref{59}) and of the fact that $p_n(x)$ is a polynomial.

\end{proof}

From this theorem the following result can be deduced:

\begin{cor}\label{cor3}
	If we take $N_{\varphi(1)}N_{\psi(1)}=\left(2\pi e^{k^2}\right)^{-1/2}$, then  $\varphi_n^{(1)}(x)\psi_m^{(1)}(x)\in\Lc^1(\mathbb{R})$ and
	\be
	\langle  \varphi_n^{(1)},\psi_n^{(1)} \rangle =\delta_{n,m},
	\label{511}\en
	for all $n,m\geq0$.
\end{cor}

\begin{proof}
First of all we observe that, because of (\ref{59}) and of the explicit forms for $\varphi_0^{(1)}(x)$ and $\psi_0^{(1)}(x)$, we have
$$
\varphi_n^{(1)}(x)\psi_m^{(1)}(x)=N_{\varphi(1)}N_{\psi(1)}p_n(x)p_m(x)e^{-x^2/2-kx}=\left(2\pi e^{k^2}\right)^{-1/2}p_n(x)p_m(x)e^{-x^2/2-kx},
$$
which is integrable for all $n$ and $m$. Therefore $\langle  \varphi_n^{(1)},\psi_m^{(1)} \rangle$ is well defined, despite of the fact that $\psi_m^{(1)}(x)\notin\Lc^2(\mathbb{R})$. Next, we show that this scalar product is zero if $n\neq m$. This can be proved by showing first the following equality, which strongly recall a similar equality for the Hermite polynomials (which are indeed quite related to our $p_n(x)$, as we will show later):
\be
(-1)^m\sqrt{m!}\,p_m(x)=e^{x^2/2+kx}\left(\frac{d^m}{dx^m}e^{-x^2/2-kx}\right),
\label{512}\en
$m=0,1,2,3,\ldots$. This equality, again, can be proved by induction, but we will not do it here.

Now, let us suppose that $m>n$. Then we have
$$
\langle  \varphi_n^{(1)},\psi_m^{(1)} \rangle=\left(2\pi e^{k^2}\right)^{-1/2}\int_{\mathbb{R}}p_n(x)p_m(x)e^{-x^2/2-kx}dx=$$
$$=
\left(2\pi e^{k^2}\right)^{-1/2}\frac{(-1)^m}{\sqrt{m!}}\int_{\mathbb{R}}p_n(x)e^{x^2/2+kx}\left(\frac{d^m}{dx^m}e^{-x^2/2-kx}\right)e^{-x^2/2-kx}dx=$$
$$=\left(2\pi e^{k^2}\right)^{-1/2}\frac{(-1)^m}{\sqrt{m!}}\int_{\mathbb{R}}p_n(x)\left(\frac{d^m}{dx^m}e^{-x^2/2-kx}\right)dx,
$$
which is zero. In fact, using integration by parts, we have to compute $\frac{d^mp_n(x)}{dx^m}$, which is zero since $m>n$. Of course, the same proof works also if $m<n$, by inverting the role of $p_n(x)$ and $p_m(x)$. 

Let us now see what happens if $n=m$. In this case, integrating by parts, we have
$$
\langle  \varphi_n^{(1)},\psi_n^{(1)} \rangle=\left(2\pi e^{k^2}\right)^{-1/2}\frac{(-1)^n}{\sqrt{n!}}\int_{\mathbb{R}}p_n(x)\left(\frac{d^n}{dx^n}e^{-x^2/2-kx}\right)dx=$$
$$=\frac{\left(2\pi e^{k^2}\right)^{-1/2}}{\sqrt{n!}}\int_{\mathbb{R}}\frac{d^np_n(x)}{dx^n}e^{-x^2/2-kx}dx.
$$
It is obvious that, for all $n\geq0$, $\frac{d^np_n(x)}{dx^n}$ is a constant. In fact, we can check that this constant is $\sqrt{n!}$: $\frac{d^np_n(x)}{dx^n}=\sqrt{n!}$. The proof is based on the following preliminary result:
\be
\frac{d^n}{dx^n}\left[p_n'(x)(x+k)\right]=np_n^{(n)}(x),
\label{514}\en
which is a direct consequence of the formula $\frac{d^n}{dx^n}(f(x)g(x)=\sum_{k=0}^{n}\left(
\begin{array}{c}
n  \\
k  \\
\end{array}
\right)f^{(n-k)}(x)g^{}(k)(x)$, that in our case reduces to just two contributions, since $(x+k)^{{(n)}}=0$ if $n\geq2$. We are now ready to prove our claim by induction on $n$. First of all, since $p_0(x)=1$, it is true that $\frac{d^0p_n(x)}{dx^0}=\sqrt{0!}$. Now, let us suppose we have $\frac{d^np_n(x)}{dx^n}=\sqrt{n!}$. We want to check that this equality extends to $n+1$. For that we use (\ref{510}). Therefore
$$
\frac{d^{n+1}}{dx^{n+1}}p_{n+1}(x)=\frac{d^{n+1}}{dx^{n+1}}\frac{1}{\sqrt{n+1}}\left(p_{n}(x)(x+k)-p_{n}'(x)\right)=\frac{1}{\sqrt{n+1}}\frac{d^{n+1}}{dx^{n+1}}p_{n}(x)(x+k),
$$
since $p_n'(x)$ is a polynomial af degree $n-1$, which is annihilated by the $n+1$-th derivative. Now, 
$$
\frac{d^{n+1}}{dx^{n+1}}p_{n+1}(x)=\frac{1}{\sqrt{n+1}}\frac{d^{n}}{dx^{n}}\left(p_{n}'(x)(x+k)+p_n(x)\right)=\frac{1}{\sqrt{n+1}}\left(np_n^{(n)}(x)+p_n^{(n)}\right)=\sqrt{(n+1)!},
$$
because of the (\ref{514}) and of our inductive assumption.

We are now ready to conclude our proof: 
$$
\langle  \varphi_n^{(1)},\psi_n^{(1)} \rangle=\left(2\pi e^{k^2}\right)^{-1/2}\int_{\mathbb{R}}e^{-x^2/2-kx}dx=1,
$$
which is what we had to prove.

\end{proof}

As we have already anticipated, the above result is not really surprising, since it is possible to deduce the following relation between our  $p_n(x)$ and the Hermite polynomials $H_n(x)$:
\be
p_n(x)=\frac{1}{2^n\,n!}H_n\left(\frac{x+k}{\sqrt{2}}\right),
\label{517}
\en
and therefore, because of (\ref{59}),
\be
\varphi_n^{(1)}=\frac{N_{\varphi(1)}}{2^n\,n!}H_n\left(\frac{x+k}{\sqrt{2}}\right)e^{-kx-e^x}, \qquad 
\psi_n^{(1)}=\frac{N_{\psi(1)}}{2^n\,n!}H_n\left(\frac{x+k}{\sqrt{2}}\right)e^{e^x-x^2/2},
\label{518}\en
where the normalization constants satisfy the condition given in Corollary \ref{cor3}. We conclude that the sets $\F_{\varphi(1)}$ and $\F_{\psi(1)}$ are biorthogonal, even if the functions in $\F_{\psi(1)}$ are not square-integrable. This suggests that $\F_{\varphi(1)}$, thought being complete\footnote{Completeness has been met previously in this paper. We remind that this means that the only square-integrable function $f(x)$ which is orthogonal to all the $\varphi_n^{(1)}(x)$'s is the zero function. This can be proved using the same argument adopted in \cite{baginbagbook}, see also \cite{kolfom}.} in $\Lc^2(\mathbb{R})$, it is not a basis. This is because, otherwise, an unique biorhogonal basis would exist in $\Lc^2(\mathbb{R})$, which is also a basis, \cite{chri}. On the other hand, here, we see that such a biorthogonal set exists, but this is not in $\Lc^2(\mathbb{R})$. Let then $\D_{\varphi(1)}=l.s.\{\varphi_n^{(1)}(x)\}$, the linear span of the functions $\varphi_n^{(1)}(x)$. This set is dense in $\Hil$, since $\F_{\varphi(1)}$ is complete. It is obvious that all functions in $\D_{\varphi(1)}$ can be written as
$$
f(x)=\sum_{k=0}^N\langle \psi_k^{(1)},f\rangle\varphi_k^{(1)}(x),
$$
for some finite $N$. Analogously, introducing the linear span of $ \psi_n^{(1)}(x)$, $\D_{\psi(1)}$, this is a subset of $\Lc^1_{loc}(\mathbb{R})$ (not dense, in general). Any function $g(x)$ in $\D_{\psi(1)}$ can be written in terms of the $\varphi_n^{(1)}(x)$:
$$
g(x)=\sum_{k=0}^M\langle \varphi_k^{(1)},g\rangle\psi_k^{(1)}(x),
$$
for some finite $M$. These formulas are not compatible with the general pseudo-bosonic framework described in \cite{baginbagbook}, while they fit well the weak pseudo-bosonic framework, \cite{bagweak}. In particular, for these $f(x)$ and $g(x)$, we deduce that $\langle f,g\rangle$ is well defined and that
\be
\langle f,g\rangle=\sum_{k=0}^{\min\{N,M\}}\langle f,\psi_k^{(1)}\rangle\langle \varphi_k^{(1)},g\rangle,
\label{519}\en
which is a sort of resolution of the identity restricted to suitable spaces. Of course, if $f(x),g(x)\in\D_{\varphi(1)}\cap\D_{\psi(1)}$, we get
$$
\langle f,g\rangle=\sum_{k=0}^{\min\{N,M\}}\langle f,\psi_k^{(1)}\rangle\langle \varphi_k^{(1)},g\rangle=\sum_{k=0}^{\min\{N,M\}}\langle f,\varphi_k^{(1)}\rangle\langle \psi_k^{(1)},g\rangle.
$$

\vspace{2mm}

An interesting question related to the example considered here is which is the role of the particular choice of the superpotentials $w_A(x)$ and $w_B(x)$ in (\ref{56}). In  fact, this is not so relevant. More in details, it is not hard to show, with the same techniques described above, that, for any choice of $w_A(x)$ and $w_B(x)$ such that $w_A(x)+w_B(x)=k+x$, $k\in\mathbb{R}$, formulas (\ref{59}) and (\ref{510}) are again satisfied. For this reason, the fact that $\varphi_n^{(1)}(x)$ or $\psi_n^{(1)}(x)$ are square-integrable or not, is related to the fact that $\varphi_0^{(1)}(x)$ or $\psi_0^{(1)}(x)$ are in $\Lc^2(\mathbb{R})$. The existence of $\langle  \varphi_n^{(1)},\psi_m^{(1)} \rangle$ can be proved in the same way, and a resolution like the one in (\ref{519}) can be again be deduced. This gives us a lot of freedom for producing examples based on the pseudo-bosonic commutation relations. We refer to \cite{bag2010} for more results on pseudo-bosons in this direction.

\section{Conclusions}\label{sectcon}

In this paper we have proposed the factorization of a given, non-Hermitian, Hamiltonian $H_1$ using two, rather than one, superpotentials. This allows, in principle, to have two vacua for $H_1$ and for its Susy counterpart, $H_2$. The eigenvectors for $H_1^\dagger$ and $H_2^\dagger$ have also been analysed, and a rich framework with two maps, the SUSY and the adjoint maps,  have been discussed in some details. We have shown how this structure can be deduced using a suitable deformation of ordinary SUSY qm. We have also discussed how bi-coherent states of the Gazeau-Klauder type can be defined, and which properties do they have. The Swanson Hamiltonian has been used as a test for our construction.

In the last part of the paper, we have discussed some applications of our general framework to the Black-Scholes Hamiltonian and to (weak) pseudo-bosonic systems.

Among other aspects, in our future research we plan to consider the role of shape-invariant potentials (if any), and analyse in more details the differences between broken and unbroken Susy. Also, we hope to refine the mathematical aspects of our approach, considering for instance what happens if the deformation operator $T$ in Section \ref{sectdeformSUSY},  its inverse, or both, are unbounded.

\section*{Acknowledgements}

This work was partially supported by the University of Palermo, and by the Gruppo Nazionale di Fisica Matematica of Indam.

\section*{Fundings}

No financial sources have to be acknowledged.

\section*{Conflicts of interest}

There are no conflicts of interest.

\section*{Availability of data and material}

Not applicable.

\section*{Code availability}

Not applicable.

\section*{Authors' contributions}

Not applicable.

\end{document}